\documentclass[cmp]{svjour}  %envcountsame
\usepackage{amsmath}
\usepackage{amsfonts,amssymb}

\journalname{Communications in Mathematical Physics}

\newlength{\Taille}

\newcommand{\re}{\mathrm{Re}}
\newcommand{\im}{\mathrm{Im}}

\newcommand{\flechebas}[1]{
  \settoheight{\unitlength}{\mbox{$#1$}}
  \settowidth{\Taille}{\mbox{~${\scriptstyle #1}$}}
  \addtolength{\unitlength}{4ex}
  \begin{picture}(0,1)
    \put(0,1){\vector(0,-1){1}}
    \put(0,0.5){\makebox(0,0){${\scriptstyle #1}$ \hspace{\the\Taille}}}
  \end{picture}}
\newcommand{\flechehaut}[1]{
  \settoheight{\unitlength}{\mbox{$#1$}}
  \settowidth{\Taille}{\mbox{~${\scriptstyle #1}$}}
  \addtolength{\unitlength}{4ex}
  \begin{picture}(0,1)
    \put(0,0){\vector(0,1){1}}
    \put(0,0.5){\makebox(0,0){\hspace{\the\Taille}${\scriptstyle #1}$ }}
  \end{picture}}
\newcommand{\flechedroite}[1]{
  \settowidth{\unitlength}{\mbox{$#1$}}
  \settoheight{\Taille}{\mbox{${\scriptstyle #1}$}}
  \addtolength{\Taille}{1ex}
  \addtolength{\unitlength}{4ex}
  \raisebox{0.5ex}{
  \begin{picture}(1,0)
    \put(0,0){\vector(1,0){1}}
    \put(0.5,0){\makebox(0,0){${\scriptstyle #1}$ \vspace{\the\Taille}}}
  \end{picture}}}
\newcommand{\flechegauche}[1]{
  \settowidth{\unitlength}{\mbox{$#1$}}
  \settoheight{\Taille}{\mbox{${\scriptstyle #1}$}}
  \addtolength{\Taille}{1ex}
  \addtolength{\unitlength}{4ex}
  \raisebox{0.5ex}{
  \begin{picture}(1,0)
    \put(1,0){\vector(-1,0){1}}
    \put(0.5,0){\makebox(0,0){${\scriptstyle #1}$ \vspace{\the\Taille}}}
  \end{picture}}}

\newcommand{\RM}{\mathbb{R}}
\newcommand{\ZM}{\mathbb{Z}}

\newcommand{\NM}{\mathbb{N}}
\newcommand{\CM}{\mathbb{C}}

\DeclareMathOperator{\sign}{sign}
\newcommand{\fracsm}[2]{\begin{matrix}\frac{#1}{#2}\end{matrix}}
\DeclareMathOperator{\vol}{vol}

\begin{document}

\title{{Dimensional asymptotics of effective actions on $S^n$, and proof of B\"{a}r-Schopka's conjecture}}
\titlerunning{Dimensional asymptotics of effective actions on $S^n$}

\author{Niels Martin M{\o}ller\inst{1}}
\institute{Department of Mathematics, University of Aarhus,
  Ny Munkegade Bld. 530, Denmark.\\ \email{nmm@imf.au.dk}}
\authorrunning{Niels Martin M\o{}ller}

\date{1 September 2007}
%\communicated{}

\maketitle
\begin{abstract}
We study the dimensional asymptotics of the effective actions, or functional determinants, for the Dirac operator $D$ and Laplacians $\Delta +\beta R$ on round $S^n$. For Laplacians the behavior depends on ``the coupling strength'' $\beta$, and one cannot in general expect a finite limit of $\zeta'(0)$, and for the ordinary Laplacian, $\beta=0$, we prove it to be $+\infty$, for odd dimensions. For the Dirac operator, B\"{a}r and Schopka conjectured a limit of unity for the determinant (\cite{BS}), i.e.
\[
\lim_{n\to\infty}\det(D, S^n_{\mathrm{can}})=1.
\]

We prove their conjecture rigorously, giving asymptotics, as well as a pattern of inequalities satisfied by the determinants. The limiting value of unity is a virtue of having ``enough scalar curvature'' and no kernel. Thus for the important (conformally covariant) Yamabe operator, $\beta=(n-2)/(4(n-1))$, the determinant tends to unity.

For the ordinary Laplacian it is natural to rescale spheres to unit volume, since
\[
\lim_{k\to\infty}\det(\Delta, S_\mathrm{rescaled}^{2k+1})=\frac{1}{2\pi e}.
\]
\end{abstract}

\section{Introduction}
\label{sec:Introduction}
Ever since spectral zeta functions of natural geometric elliptic differential operators on manifolds started appearing in mathematics (\cite{RS1}, \cite{RS2}) and in the regularization of path integrals (\cite{Ha}), there has been interest in calculating values of the associated zeta determinants. Explicit formulae on various spaces have been derived (\cite{BrFuncDet}, \cite{BS}, \cite{DK}, \cite{Dow1}, \cite{Dow2}, \cite{Ki}, \cite{Va}, \cite{We1}, \cite{We2}), and some authors have displayed numerical features of the behavior of such determinants, culminating in the formulation of Conjecture 1 in \cite{BS}, concerning the limit of the determinants of the Dirac operator on standard $n$-dimensional spheres, as the dimension $n$ grows large. Thus it seems an opportune time for investigating rigorously the asymptotics, using the previously established formulae for determinants. Contrary to what one might think, the rather complicated looking series, which involve Barnes zeta functions and generalized Bernoulli and Stirling numbers, can often be understood via quite elementary methods.

In Section \ref{sec:Dirac} we find asymptotics of the effective action for the Dirac operator, and of the phase factor sometimes included. As a corollary this resolves the questions raised in the form of Conjecture 1 in \cite{BS}. The key point in the proof is the use of
explicit formulas by Thomas Branson, B\"{a}r-Schopka and J. S. Dowker, and application of recursion to control the indirectly defined special functions and polynomial coefficients that appear, along with the functional equation for the Riemann zeta function $\zeta_R$ and some elementary estimates of $\zeta_R$ (\cite{BrFuncDet}, \cite{BS}, \cite{Dow1}, \cite{Dow2}. See also the book by K. Kirsten
\cite{Ki}).

In Section \ref{sec:Laplace} for the Laplace-Beltrami operator on functions, more refined methods are
needed. The reason seems to be that the ordinary Laplacian is not a ``natural object'' to consider. While an approach using polynomial
coefficients such as (generalized) Stirling and Bernouilli numbers works well for the more natural conformally covariant Yamabe and Dirac operators, where in fact the zeta function and derivative at zero tend rapidly to zero, then in the cases when it tends to $\infty$, formidable cancellations amongst large terms should occur in such sum formulas. Namely for Dirac and Yamabe the convergence holds with absolute signs inside the Bernoulli/Stirling expansion sums, while for the ordinary Laplacian they do certainly not. As shown below, the zeta derivative at zero of the ordinary Laplacian goes to $+\infty$ as the logarithm of the dimension, while the same quantity for Dirac and Yamabe operators converges exponentially to zero.

The proof in the ordinary Laplacian case instead relies on estimating
a broad range of Barnes zeta functions through convenient contour
integral representations. Namely, by exploiting geometry of the
situation it is possible to deform the basic Hankel contour to make
use of a hybrid of Laplace's method and the method of stationary phase, the idea of which is exactly
that of cancellation of increasingly rapid oscillations. We note that
the standard formulations of the method of stationary phase and
Laplace's method for contour integrals (see for example \cite{Ol})
do not seem feasible for the estimates needed here. What is needed is a mix of these standard tools, due to phase factors which are neither real nor purely imaginary. In this respect
the present paper may also turn out useful for treating similar problems on contour integrals.

\section{Dirac operator sphere determinants}
\label{sec:Dirac}
As in \cite{BS}, we define the Dirac operator determinant including a phase as follows.
\begin{definition}
\begin{equation}
\det(D):=\exp\Big(i\frac{\pi}{2}\big(\zeta_{D^2}(0)-\eta_D(0)\big)\Big)\exp\Big(-\frac{1}{2}\zeta^{\,\prime}_{D^2}(0)\Big).
\end{equation}
\end{definition}
To state the theorem, we fix the notation that $\varphi_D=\frac{\pi}{2}\zeta_{D^2}(0)-\eta_D(0)$. The theorem in particular proves Conjecture 1 of \cite{BS}.

\begin{theorem}
For the standard round spheres we have:
\begin{equation}
\lim_{n\to\infty}\det(D, S^n)=1.
\end{equation}
In fact we have the following asymptotics
\[
\begin{split}
|\det(D, S^n)|&=\exp\Big(O\big((\fracsm{3}{4})^n\big)\Big),\\
\varphi_D&=O\big((\fracsm{3}{4})^n\big).
\end{split}
\]
Furthermore we have the inequalities
\begin{equation}
\begin{split}
&|\det(D, S^n)|< 1,\quad\varphi_D>0,\quad\textrm{if}\quad n\equiv0\:\:\textrm{(mod 4)},\\
&|\det(D, S^n)|> 1,\quad\varphi_D=0,\quad\textrm{if}\quad n\equiv1\:\:\textrm{(mod 4)},\\
&|\det(D, S^n)|> 1,\quad\varphi_D<0,\quad\textrm{if}\quad n\equiv2\:\:\textrm{(mod 4)},\\
&|\det(D, S^n)|< 1,\quad\varphi_D=0,\quad\textrm{if}\quad n\equiv3\:\:\textrm{(mod 4)}.
\end{split}
\end{equation}
\end{theorem}
\begin{remark}
This pattern is also visible in the numerics in \cite{BS} and is interesting in comparison to the present author's results on local extremals of determinants, where that pattern is found to be (max, max, min, min), again depending on the respective dimensions mod $4$ (see \cite{Moeller}). This is yields another example of mod $4$ dependencies generically showing up for zeta regularised quantities.
\end{remark}
\begin{proof}
We need to show that, writing from now on $S^n$ for $S^n_{\mathrm{can}}$,
\begin{equation}
\begin{split}
&\lim_{n\to\infty}\zeta_{(D^2, S^n)}(0)=0,\\
&\lim_{n\to\infty}\zeta_{(D^2, S^n)}^{\,\prime}(0)=0.
\end{split}
\end{equation}

We firstly review the expressions for the zeta functions of Dirac squared, by Tom Branson \cite{BrFuncDet}. Depending on the parity of $n$ we have as follows, in perfect agreement with \cite{BS}.
\begin{equation}\label{DiracReview}
\begin{split}
&\zeta_{(D^2,S^{n})}(s)=\frac{2^{k+1}}{(2k-1)!}\sum_{\alpha=0}^{k-1}d_{\alpha,k}\zeta_R(2s-2\alpha-1),\quad\textrm{if}\quad n=2k,\\
&\zeta_{(D^2,S^{n})}(s)=\frac{2^{k}}{(2k-2)!}\sum_{\alpha=0}^{k-1}e_{\alpha,k}\big(2^{2s-2\alpha}-1\big)\zeta_R(2s-2\alpha),\quad\textrm{if}\quad n=2k-1,
\end{split}
\end{equation}
where the $d_{\alpha,k}$ and $e_{\alpha,k}$ are integers defined indirectly through the following polynomial expressions in $x$.
\begin{equation}\label{deDef}
\prod_{p=1}^{k-1}\big(x-p^2\big)=\sum_{\alpha=0}^{k-1}d_{\alpha,k}x^\alpha,\quad\textrm{and}\quad
\prod_{p=1}^{k-1}\big(x-(p-\fracsm{1}{2})^2\big)=\sum_{\alpha=0}^{k-1}e_{\alpha,k}x^\alpha.
\end{equation}
Also, to go from Branson's formulas, we applied the functional equation
\begin{equation}
\zeta_{1/2}(s)=\big(2^s-1\big)\zeta_R(s),
\end{equation}
since it turns out convenient to express everything in terms of the
Riemann zeta function $\zeta_R$.

We shall repeatedly rely on the controllability of $\zeta_R$ along the
positive real axis. In fact we shall apply only the following two quite elementary facts.
\begin{align}\label{ZetaEstimates}
&0\leq\zeta_R(n)\leq C_{R},\quad n\in\NM_+,\\ \label{ZetaPrimeEstimates}
&0\leq\zeta_R^{\,\prime}(n)\leq C_{R}',\quad n\in\NM_0,
\end{align}
for suitably chosen constants.

If now $n$ is odd, we have as always $\zeta_{(D^2, S^n)}(0)=0$, since the kernel is trivial. Thus we let $n$ be even and apply (\ref{DiracReview}). Recalling the functional equation for the Riemann zeta function, in the form most convenient here,
\begin{equation}\label{FuncEq}
\zeta_R(s)=2^s\pi^{s-1}\sin\Big(\frac{\pi s}{2}\Big)\Gamma(1-s)\zeta_R(1-s),
\end{equation}
we may now rewrite in a form where the asymptotics are more readily seen.
\[
\zeta_{(D^2,S^{2k})}(0)=\frac{2^{k+2}}{(2k-1)!}\sum_{\alpha=0}^{k-1}d_{\alpha,k}(-1)^{\alpha+1}(2\pi)^{-2\alpha-2}(2\alpha+1)!\zeta_R(2\alpha+2).
\]
Note that for fixed $k$ this is a sum in which all the terms have the same sign. Namely
\[
\sign\big[d_{\alpha,k}\big]=(-1)^{k-1-\alpha},
\]
while $\zeta_R$ is positive, since $\zeta_R(0)$ ($=-\frac{1}{2}$) is excluded. One thing we get from this is
\[
\sign\big[\zeta_{(D^2,S^{2k})}(0)\big]=(-1)^{k}.
\]
Very importantly the constant sign (for fixed $k$) allows us to estimate term by term and we write
\begin{equation}
\big\vert\zeta_{(D^2,S^{2k})}(0)\big\vert\leq C_RA(k),
\end{equation}
where the positive numbers $A(k)$ are defined as
\begin{equation}\label{ADef}
A(k):=\frac{2^{k+2}}{(2k-1)!}(-1)^{k+1}\sum_{\alpha=0}^{k-1}d_{\alpha,k}(-1)^\alpha(2\pi)^{-2\alpha-2}(2\alpha+1)!.
\end{equation}
The point is now that we can control the indirectly defined polynomial coefficients $d_{\alpha,k}$ through a simple recursion. From (\ref{deDef}) we get
\begin{equation}
d_{\alpha,k+1}=d_{\alpha-1,k}-k^2d_{\alpha,k}.
\end{equation}
Note that it is of course implicitly understood that $d_{\alpha,k}=0$, if $k$ is not in the range $\{0,1,\ldots, k-1\}$.
By a change of index for the first term and using $2\alpha(2\alpha+1)\leq 2k(2k+1)$, this shows the following estimate.
\begin{equation}
\begin{split}
A(k+1)&\leq\frac{2}{2k(2k+1)}\bigg\{k^2+\frac{2k(2k+1)}{(2\pi)^2}\bigg\}A(k)\\
&=\bigg\{\frac{k}{2k+1}+\frac{1}{2\pi^2}\bigg\}A(k)\\
&\leq\frac{5}{9}A(k).
\end{split}
\end{equation}
Thus for instance with $\delta^2=\frac{5}{9}<1$, we get
\[
\zeta_{(D^2,S^{n})}(0)=O\big(\delta^n\big),
\]
proving in particular the claim that
\begin{equation}\label{zetazero}
\lim_{n\to\infty}\zeta_{(D^2, S^n)}(0)=0,
\end{equation}
meaning that the phase converges to zero. Furthermore note that with $\delta^2=0.55\ldots$ this convergence is indeed rapid, as also witnessed by the numerics of B\"{a}r-Schopka.

To deal with the derivative we differentiate (\ref{DiracReview}) at $s=0$ and use $\zeta(-2\alpha)=0,\: \alpha\in\NM_{+}$.
\begin{align}
&\zeta_{(D^2,S^{2k})}^{\,\prime}(0)=\frac{2^{k+2}}{(2k-1)!}\sum_{\alpha=0}^{k-1}d_{\alpha,k}\zeta_R^{\,\prime}(-2\alpha-1),\\ \label{OddDeriv}
&\zeta_{(D^2,S^{2k-1})}^{\,\prime}(0)=\frac{2^{k+1}}{(2k-2)!}\bigg\{-e_{0,k}\log 2+\sum_{\alpha=0}^{k-1}e_{\alpha,k}\big(2^{-2\alpha}-1\big)\zeta_R^{\,\prime}(-2\alpha)\bigg\}.
\end{align}
Again we use the functional equation (\ref{FuncEq}) which by differentiation gives
\begin{equation}
\begin{split}\label{evenderiv}
&\zeta_R^{\,\prime}(-2\alpha-1)=2(-1)^\alpha(2\pi)^{-2\alpha-2}(2\alpha+1)!\zeta_R^{\,\prime}(2\alpha+2),\quad\alpha\in\NM_0,\\
&\quad\quad\quad\quad\quad\quad\quad+\big[\log(2\pi)+\gamma-H_{2\alpha+1}\big]\zeta_R(-2\alpha-1),
\end{split}
\end{equation}
\begin{equation}
\zeta_R^{\,\prime}(-2\alpha)=\pi(-1)^\alpha(2\pi)^{-2\alpha-1}(2\alpha)!\zeta_R(2\alpha+1),\quad\alpha\in\NM_+.
\end{equation}
Here $\gamma$ is Euler's constant, and we have applied
\begin{equation}\label{GammaDiff}
\frac{\Gamma^{\,\prime}(n)}{\Gamma(n)}=H_{n-1}-\gamma,\quad H_{n-1}=\sum_{j=1}^{n-1}\frac{1}{j},\quad n\in\NM.
\end{equation}

Dealing first with the even dimensional case $n=2k$, we note that the second term in (\ref{evenderiv}) gives a contribution that converges to zero. This follows since in absolute value, after summing over $\alpha$, it is altogether dominated by
\[
\Big\vert\big(\log(2\pi)+\gamma+H_{2k+1}\big)\zeta_{(D^2, S^{2k})}(0)\Big\vert=O\big(\delta^{2k}\big),
\]
where again $\delta^2=\frac{5}{9}<1$ for example is admissible, since $\{\frac{9}{10}(1+1/\pi^2)\}^k H_{2k+1}$ is bounded.

To control now the first term in (\ref{evenderiv}), note with $A(k)$ from (\ref{ADef}) we have similar estimates here
\[
\bigg\vert\frac{2^{k+2}}{(2k-1)!}\sum_{\alpha=0}^{k-1}d_{\alpha,k}\zeta_R^{\,\prime}(-2\alpha-1)\bigg\vert\leq C_R'A(k),
\]
thus we have the desired convergence
\[
\lim_{k\to\infty}\zeta_{(D^2, S^{2k})}^{\,\prime}(0)=0,
\]
and in fact $\zeta_{(D^2, S^{2k})}^{\,\prime}(0)=O\big(\delta^{2k}\big)$ with, for example, $\delta^2=\frac{5}{9}$ again.

In odd dimensions the situation is changed only slightly. The polynomial coefficients are now $e_{\alpha,k}$ from (\ref{deDef}) and as before we have
\begin{equation}
\sign[e_{\alpha,k}]=(-1)^{k-1-\alpha}.
\end{equation}
Writing now (\ref{OddDeriv}) as
\[
\zeta_{(D^2,S^{2k-1})}^{\,\prime}(0)=-\frac{2^{k+1}}{(2k-2)!}\sum_{\alpha=0}^{k-1}e_{\alpha,k}Z(-2\alpha),
\]
with
\begin{equation}
Z(-2\alpha)= \begin{cases}
\big(1-2^{-2\alpha}\big)\zeta_R^{\,\prime}(-2\alpha) & \alpha\in\NM_+\\
\log 2 & \alpha=0.
\end{cases}
\end{equation}
Note that once again the signs match up
\[
\sign[Z(-2\alpha)]=(-1)^{k-1-\alpha},
\]
and for a constant $\tilde{C}_R:=\max\big(\log 2,\fracsm{C_R}{2}\big)$ the estimates
\[
\big\vert Z(-2\alpha)\big\vert\leq C(2\pi)^{-2\alpha}(2\alpha)!
\]
hold, and we define a new sequence of positive numbers $B(k)$ by
\[
\big\vert\zeta_{(D^2,S^{2k-1})}^{\,\prime}(0)\big\vert\leq B(k):=\frac{\tilde{C}_R2^{k+1}}{(2k-2)!}(-1)^{k-1}\sum_{\alpha=0}^{k-1}e_{\alpha,k}(-1)^\alpha(2\pi)^{-2\alpha}(2\alpha)!.
\]
Again we find recursion relations for the relevant polynomials
\begin{equation}
e_{\alpha,k+1}=-(k-\fracsm{1}{2})^2e_{\alpha,k}+e_{\alpha-1,k},
\end{equation}
which gives recursive estimates on the $B(k)$
\[
B(k+1)\leq\bigg\{\frac{1}{2}+\frac{1}{2\pi^2}\bigg\}B(k)\leq\frac{5}{9}B(k),
\]
and this concludes the proof of the theorem.
\end{proof}

\section{Laplace operator sphere determinants}
\label{sec:Laplace}
Following \cite{Dow1} we investigate Laplacians of the form
\begin{equation}
L_{\alpha}=\Delta +\frac{n-1}{4n}R-\alpha_n^2,
\end{equation}
where $R$ is the scalar curvature (of $S^n$) and the $\alpha_n\in\RM$ are constants, but may depend on the dimension $n$. We restrict here to the case $0\leq \alpha_n\leq \frac{n-1}{2}$. Recalling now that
\[
R(S^n_{\mathrm{can}})=n(n-1),
\]
we see that $\alpha_n=\frac{n-1}{2}$ gives the ordinary Laplacian, while $\alpha_n=\frac{1}{2}$ corresponds to the Yamabe operator in each dimension. Note that there is no kernel of the operator $L$ in dimension $n$ if $\alpha_n<\frac{n-1}{2}$, and that for $\alpha_n=\frac{n-1}{2}$ the kernel is the constant functions on $S^n$.

The theorem we will prove in this section is the following.
\begin{theorem}\label{LaplaceTheorem}
On the standard (radius of unity) spheres we have the limits
\begin{equation}
\lim_{k\to\infty}\det(\Delta, S^{2k+1}_{\mathrm{can}})=0.
\end{equation}
and
\begin{equation}
\lim_{n\to\infty}\det(Y, S^n_{\mathrm{can}})=1
\end{equation}
In fact, we can display the asymptotics as (for $n=2k+1$)
\[
\zeta'_{\Delta, S_\mathrm{can}^{n}}(0)=\log n+O\Big(\frac{\log\log n}{\log n}\Big),
\]
and (for any $n$)
\[
\zeta'_{Y, S^{n}_{\mathrm{can}}}(0)=O\big(2^{-n}\big).
\]
\end{theorem}
\begin{remark}
The proof of the claimed limit for the conformal Laplacian is easily
carried out analogously to the Dirac case. However we give a different type of proof, which works for the ordinary Laplacian as
well, and is more suitable for exposing the significance of the coupling constant and kernel of the operator.
\end{remark}

From \cite{Dow1} and \cite{Dow2} we get the corresponding zeta functions and their first derivatives at zero
\begin{equation}\label{DowkerFormula}
\begin{split}
\zeta_L'(0)=&\zeta_n'(0,a_--\alpha_n)+\zeta_n'(0,a_+-\alpha_n)+\zeta_n'(0,a_-+\alpha_n)+\zeta_n'(0,a_++\alpha_n)\\
&[+\ln(n-1)]\:-\sum_{j=1}^{\lfloor\frac{n}{2}\rfloor}\frac{\alpha_n^{2j}}{j}N_{2j}(n)\sum_{i=0}^{j-1}\frac{1}{2i+1}.
\end{split}
\end{equation}
In the notation here, the numbers $a_\pm:=(n\pm 1)/2$ originate in the contributions from Dirichlet ($a_+$) and Neumann ($a_-$) boundary conditions on hemispheres. For the ordinary Laplacian, the four functions $a_{\pm}\pm\alpha_n$ are, by order of appearance in the above equation, $(0,1,n-1,n)$. For the Yamabe operator it reads $(\fracsm{n}{2}-1,\fracsm{n}{2},\fracsm{n}{2},\fracsm{n}{2}+1)$.

The notation ``$[+\ln(n-1)]$'' means that this term is only present if $L_{\alpha}$ has non-trivial kernel. The numbers $N_{2j}(n)$ are defined using the following polynomials, namely let
\[
\prod_{p=1}^{n-2}\Big(x+\frac{n-1}{2}-p\Big)=\sum_{r=0}^{n-2}x^r\tilde{N}_r(n).
\]
Then
\[
N_{2j}(n):=\frac{2}{(n-1)!}\tilde{N}_{2j-2}(n).
\]
Finally the zeta functions appearing here are (special cases of) Barnes zeta functions, where for $a\in\RM$
\begin{equation}\label{Barnes}
\zeta_n(s,a)=\sum_{m=0}^{\infty}\binom{m+n-1}{n-1}\frac{1}{(a+m)^s},\quad\re{s}>n,
\end{equation}
Note that here, if $a=0$, it is implicit that the term $m=0$ is omitted.

The meromorphic continuation is carried out using the contour integrals
\begin{equation}\label{BarnesIntegral}
\zeta_n(s,a)=\frac{i\Gamma(1-s)}{2\pi}\int_H\frac{e^{az}}{(1-e^{z})^n}z^{s-1}dz,
\end{equation}
where $H$ is a left Hankel contour, and $z^{s-1}$ is defined using the
negative real axis branch cut of the logarithm enclosed by $H$. The
contours will be chosen conveniently for getting estimates of the
Barnes functions. Deformations of the curves must of course respect the branch cut and the possible poles at $2\pi i k,k\in\ZM$.

We note a few general features of the Barnes zeta
functions. Each contour will be taken as the positively oriented
boundary of some box
$(-\infty,r_n]\times[-\fracsm{i\pi}{2},\fracsm{i\pi}{2}]$, where $r_n$
is a real-valued function of the dimension $n$. We denote such a contour by $\gamma_n$ and decompose in the obvious way into one vertical and two horizontal components $h_n^+\cup v_n\cup h_n^-$. Now, since
\begin{equation}\label{NormofDenom}
\big\vert 1-e^{z}\big\vert^n=\big(1+e^{2x}-2e^{x}\cos(y)\big)^{\frac{n}{2}}
\end{equation}
is the norm of the denominator at $z=x+iy$, we always have
\begin{equation}\label{NormOnPiHalves}
\big\vert 1-e^{z}\big\vert^n=\big(1+e^{2x}\big)^{\frac{n}{2}}
\end{equation}
on the horizontal parts of the contours. From (\ref{BarnesIntegral})
we see, given that $a\leq n$,
\begin{equation}\label{HorizontalParts}
n^k\Bigg\vert\frac{e^{az}}{(1-e^z)^n}\Bigg\vert\leq n^k\big(1+e^{-2x}\big)^{-\fracsm{n}{2}}\to 0,\quad\text{for}\quad n\to\infty,
\end{equation}
for fixed $z$ and $k$, i.e. pointwise convergence to zero of the
integrands. Lebesgue dominated convergence theorem can therefore in most case be used to prove that the contribution from a horizontal piece is $O(n^{-\infty})$ as $n\to\infty$, meaning by definition $O(n^{-k})$ as $n\to\infty$, for any $k$.

When $a>0$, the derivatives at $s=0$ are, using (\ref{GammaDiff})
\begin{equation}\label{Derivatives}
\zeta_n'(0,a_n)=\frac{i}{2\pi}\int_{\gamma_n}\frac{e^{a_nz}}{(1-e^{z})^n}\frac{\log z+\gamma}{z}dz.
\end{equation}

\begin{proposition}
If $a_n=0,1,\fracsm{n-1}{2}$ or $\fracsm{n}{2}$, then
\[
\zeta_n'(0,a_n)\to 0,\quad\textrm{as}\quad n\to\infty,
\]
in fact it is $O(n^{-\infty})$ as $n\to\infty$.
\end{proposition}

\begin{proof}
For dealing with these cases we use the fixed contour consisting of
the positively oriented boundary of the box $(-\infty,\log
4]\times[-\fracsm{i\pi}{2},\fracsm{i\pi}{2}]$. From the general
remarks following (\ref{HorizontalParts}), the horizontal parts of the
contours all contribute $O(n^{-\infty})$ for $a_n>0$, by Lebesgue dominated convergence.

In the case $a_n=0$ however, note that the continuation of the
integral doesn't work directly around $s=0$. From (\ref{Barnes})
we deduce, by straightforward resummation, a family of functional equations
\begin{equation}\label{BarnesFuncEq}
\zeta_n(s,0)=\zeta_n(s,1)+(n-1)\zeta_n(s+1,1),
\end{equation}
thus allowing again for meromorphic continuation to $\CM$ by integral representations. Taking into account the front factor of $\Gamma(1-s)$ in (\ref{BarnesIntegral}) and using
\[
\Gamma(1-s)=\frac{1}{1-s}-\gamma+\ldots
\]
near $s=1$, shows
\begin{equation}\label{DerivaZero}
\zeta'_n(s,0)=\zeta'_n(0,1)+\frac{n-1}{2\pi i}\int_H\frac{e^z}{(1-e^z)^n}\big(\gamma\log z+\fracsm{1}{2}\log^2z\big)dz.
\end{equation}
On the vertical line segments $v_n$ we use $\vert1-e^z\vert\geq3$ and estimate
\[
\begin{split}
&(n-1)\bigg\vert\int_{v_n}\frac{e^z}{(1-e^z)^n}P(z)\;dz\bigg\vert\leq C_P(n-1)\:3^{-n}\to 0,\quad\mathrm{as}\quad n\to\infty,
\end{split}
\]
where $P$ is any polynomial in $\frac{1}{z}$ and $\log z$, as in the expressions appearing in (\ref{Derivatives}) and (\ref{DerivaZero}).

In the remaining cases $a_n=\fracsm{n-1}{2}, \fracsm{n}{2}$ we focus for the sake of argument on the latter, i.e.
\[
\int_H\frac{1}{\big(e^{z/2}-e^{-z/2}\big)^n}\frac{\log z+\gamma}{z}dz,
\]
Then using $\big\vert e^{z/2}-e^{-z/2}\big\vert^2=e^x+e^{-x}-2\cos(y)$
we see that indeed, by Lebesgue dominated convergence, there is rapid convergence to zero, namely
\[
\zeta_{(Y,S^n)}'(0)= O\big(2^{-n}\big),
\]
concluding the proof.
\end{proof}
The next proposition deals with the most complicated of the
terms, with parameters $a_n=n-1$ and $n$. Here there is no way to deform the Hankel contour in the complex plane, respecting the poles and branch cut of $\log z$, so as to obtain pointwise convergence to zero of the integrand as $n\to 0$ everywhere along the curve. This follows from (\ref{NormofDenom}), since it would imply $x\leq 0$ at some point of the curve, which would thus intersect the logarithmic branch cut.

The following lemma essentially gives a suitable hybrid of Laplace's method and the method of
stationary phase, applicable to the cases needed here.
\begin{lemma}[Laplace's method/stationary phase hybrid]\label{stationaryphaselemma}
Assume $\varphi:[a,b)\to\RM_{\geq 0}$ for $0\leq a< b<\infty$ satisfies
\begin{itemize}
\item[(i)]{} $\varphi$ is continuous and bounded,
\item[(ii)]{} $x\mapsto e^x\cdot\varphi(x)$ is increasing on $[a,b)$.
Then, writing $s_b:=\tan^{-1}(e^{-b})$,
\[
\Bigg\vert\int_a^b\big(ie^{-x}-1\big)^{-n}\varphi(x)dx\Bigg\vert\leq \frac{2\pi}{n\tan s_b}\max_{s_b\leq x\leq s_b+\fracsm{\pi}{n}}\varphi\big(\log\big(\fracsm{1}{\tan x}\big)\big)+O(n^{-\infty}),
\]
where the error term is uniform in $b$, but not generally in $a$.
\end{itemize}
\end{lemma}
\begin{remark}
\begin{itemize}
\item[(i)] Note that in applying this lemma, we will let $b$ depend on $n$, while $a$ will be fixed.
\item[(ii)] If $\varphi$  is $C^1$, the second condition in the lemma is equivalent to $\varphi'\geq-\varphi$.
\end{itemize}
\end{remark}

\begin{proposition}\label{nProposition}
If $a_n=n-1$ or $n$, then
\[
\zeta_n'(0,a_n)\to 0,\quad\text{for}\quad n\to\infty.
\]
\end{proposition}
\begin{proof}
For proving this, we assume $n\geq 4$ and shift for each $n$ the contour to the positively oriented boundary of the box $(-\infty,\log n]\times[-\fracsm{i\pi}{2},\fracsm{i\pi}{2}]$, which still encloses only the singularity at $z=0$. 

The contributions from $h^{\pm}_n$ with, say, $\re z\leq \log 4$, are
again $O(n^{-\infty})$ by (\ref{HorizontalParts}), while on the moving right hand edges $v_n$, we have the estimates
\[
\big\vert e^{-z}-1\big\vert^{-n}=(1+\fracsm{1}{n^2}-\fracsm{2}{n}\cos(y))^{-\fracsm{n}{2}}\leq\big(1-\fracsm{1}{n}\big)^{-n},
\]
yielding, for some constant $C>0$
\[
\bigg\vert\int_{v_n}\frac{1}{(e^{-z}-1)^n}\frac{\log z+\gamma}{z}dz\bigg\vert\leq C\frac{\log\log n}{\log n}.
\]
Thus the contribution from this part tends to zero as $n\to\infty$, and similarly for $a_n=n-1$, again with or without the $\log z$ present, as needed for the terms in (\ref{Derivatives}).

For the contributions from the right half-plane part of the horizontals, denoted by $h_{R,n}^{\pm}$, in the cases $a_n=n-1,n$, any finite piece gives an $O(n^{-\infty})$ contribution as $n\to\infty$.

To deal with the case $a_n=n$, we calculate explicitly,
\[
\begin{split}
&\frac{i}{2\pi}\int_{h_{R,n}^{+}\cup h_{R,n}^{-}}\big(e^{-z}-1\big)^{-n}\frac{\log z+\gamma}{z}dz=\\
&-\frac{1}{\pi}\:\im\int_{0}^{\log n}\big(ie^{-x}-1\big)^{-n}\frac{\big(x-i\pi/2\big)\Big(\log\sqrt{x^2+(\frac{\pi}{2})^2}+i\tan^{-1}(\fracsm{\pi}{2x})+\gamma\Big)}{x^2+(\fracsm{\pi}{2})^2}dx.
\end{split}
\]
To take the oscillation into account, we repeatedly apply Lemma \ref{stationaryphaselemma} to the expression, with $s_b=\fracsm{1}{n}$, considering the product terms separately. For each term the relevant positive function $\varphi$ is decreasing for $x$ large (i.e. for $n$ large, since finite pieces contribute $O(n^{-\infty})$), so the maximum is evaluation at $s_b+\frac{\pi}{n}$. For example one term is analyzed as follows
\[
\varphi_1(x):=\frac{x\log\sqrt{x^2+(\frac{\pi}{2})^2}}{x^2+(\frac{\pi}{2})^2},\quad x\in[1,\log n),
\]
\[
\begin{split}
&\Bigg\vert\int_{0}^{\log n}\big(ie^{-x}-1\big)^{-n}\varphi_1(x)dx\Bigg\vert\leq2\pi\frac{\fracsm{1}{n}}{\tan\fracsm{1}{n}}\varphi_1\Big(\log\big(\fracsm{1}{\tan(\frac{1+\pi}{n})}\big)\Big)+O(n^{-\infty}),
\end{split}
\]
so that by Lemma \ref{stationaryphaselemma}, the righthand side converges to zero, as the remaining terms can
similarly be shown to do.

In the case $a_n=n-1$, Lemma \ref{stationaryphaselemma} is not needed, since after a partial integration it is easily seen that
\[
\int_{\fracsm{i\pi}{2}+[0,\log n]}\big(e^{-z}-1\big)^{-n}e^{-z}\frac{\log
  z+\gamma}{z}dz=O(n^{-1})\quad\text{as}\quad n\to\infty,
\]
which ends the proof of Proposition \ref{nProposition}.
\end{proof}
\begin{remark}
Note that the use of Lemma \ref{stationaryphaselemma} in the proof was essential, since without the oscillating factor, amounting approximately to $\sin(ne^{-x})$, we would have
\[
\begin{split}
\frac{1}{2}\big(\log(\log n)\big)^2\geq \int_{\log\sqrt{n}}^{\log n}\big(1+e^{-2x}\big)^{-\frac{n+1}{2}}\frac{\log x}{x}dx\geq\frac{\log 2}{4}\Big\{\log(\log n)-\frac{1}{2}\log 2\Big\},
\end{split}
\]
and thus convergence to $+\infty$, at a rate between $\log\log n$ and $\log^2\log n$.
\end{remark}
We need to deal with the last term in (\ref{DowkerFormula}), which is controlled as follows.
\begin{proposition}
If $n$ is odd, then
\[
\sum_{j=1}^{\lfloor\frac{n}{2}\rfloor}\frac{\alpha_n^{2j}}{j}N_{2j}(n)\sum_{i=0}^{j-1}\frac{1}{2i+1}=0.
\]
Setting $\alpha_n=\fracsm{1}{2}$, then for $n$ of any parity
\[
\lim_{n\to\infty}\sum_{j=1}^{\lfloor\frac{n}{2}\rfloor}\frac{\alpha_n^{2j}}{j}N_{2j}(n)\sum_{i=0}^{j-1}\frac{1}{2i+1}=0,
\]
being $O(2^{-n})$ as $n\to\infty$.
\end{proposition}
\begin{remark}
Interestingly enough, this term is exactly the ``correction term''
from a product formula of certain determinants discussed in
\cite{Dow1}, and this makes it particularly interesting to see that
this tends to zero. The convergence is exponential and this means that
the anomaly, in this situation quickly evaporates as $n\to\infty$.
\end{remark}
\begin{proof}
If $n=2k+1$ is odd we exploit symmetry and rewrite the left-hand-side in the defining equation for $\tilde{N}_{r}(n)$
\[
x\prod_{p=1}^{k-1}\Big(x^2-p^2\Big)=\sum_{r=0}^{2k-1}\tilde{N}_{r}(n)x^{r},
\]
showing the vanishing of all terms with even $r$.

If $n=2k$ is even we again use symmetry to write
\[
\prod_{p=1}^{k-1}\Big(x^2-(p-\fracsm{1}{2})^2\Big)=\sum_{r=0}^{k-1}\tilde{N}_{2r}(2k)x^{2r}.
\]
This shows that $\sign[\tilde{N}_{2r-2}(2k)]=(-1)^{k+r}$ and gives as usual a recurrence relation
\[
\tilde{N}_{2r-2}(2k+2)=\tilde{N}_{2r-4}(2k)-(k-\fracsm{1}{2})^2\tilde{N}_{2r-2}(2k).
\]
We estimate
\[
\bigg\vert\sum_{j=1}^{k}\frac{\alpha_{2k}^{2j}}{j}N_{2j}(2k)\sum_{i=0}^{j-1}\frac{1}{2i+1}\bigg\vert\leq
2D(k),\quad D(k):=\sum_{j=1}^{k}\alpha_{2k}^{2j}(-1)^{k+j}\frac{\tilde{N}_{2j-2}(2k)}{(2k-1)!}.
\]
Letting $\alpha_n=\fracsm{1}{2}$ we find the recursive estimates
\begin{equation}
\begin{split}
D(k+1)&=\sum_{j=1}^{k+1}\big(\fracsm{1}{2}\big)^{2j}(-1)^{k+j}\frac{\tilde{N}_{2j-2}(2k+2)}{(2k+1)!}\\
&=\frac{1}{2k(2k+1)}\bigg\{\frac{1}{4}+(k+\fracsm{1}{2})^2\bigg\}D(k)\leq \frac{1}{2}D(k),
\end{split}
\end{equation}
for sufficiently large $k$, proving exponential convergence to zero and ending the proof.
\end{proof}

\section{Determinants on rescaled spheres}
\label{sec:rescaled}
It may be argued that taking the spheres with the radius unity
standard metrics is not as natural for the problem of determinants, and
that rescaling to unit volume is more appropriate. For the conformally
covariant Yamabe and Dirac operators, this is of course produces no
change, and focus remains on the ordinary Laplacian. Here the rescaling indeed cancels the leading term $\log n$ in the zeta derivative at zero.

If we rescale the metric $g$ by the constant $\lambda>0$ to $\tilde{g}:=\lambda^2g$ the determinant of the ordinary Laplacian changes as
\begin{equation}\label{DetChange}
\zeta'_\lambda(0)=\zeta'(0)+2\log\lambda\cdot\zeta(0)=\zeta'(0)-2\log\lambda,
\end{equation}
where the last equality is due to $\zeta(0)=-1$. To rescale the spheres we use
\[
\lambda(n)=\big(\vol(S^n)\big)^{-\frac{1}{n}}=\Bigg(\frac{2\pi^{\frac{n+1}{2}}}{\Gamma(\frac{n+1}{2})}\Bigg)^{-\frac{1}{n}},
\]
which with Stirling's formula
\[
\sqrt{2\pi}n^{n+\fracsm{1}{2}}e^{-n+\fracsm{1}{12n+1}}<n!<\sqrt{2\pi}n^{n+\fracsm{1}{2}}e^{-n+\fracsm{1}{12n}}
\]
gives
\[
\begin{split}
&\log\lambda(n)\\
&\quad=-\frac{1}{n}\Big\{\frac{n+1}{2}\log\pi-\fracsm{1}{2}\log\fracsm{\pi}{2}-\frac{n}{2}\log\frac{n-1}{2}+\frac{n-1}{2}+\frac{1}{6(n-1)}\Big\}+O(n^{-3})\\
&\quad=\frac{1}{2}\log n-\frac{1}{2}(1+\log(2\pi))+O(n^{-1}).
\end{split}
\]
Inserting in (\ref{DetChange}) gives
\[
\zeta'_{\mathrm{rescaled}}(0)=\zeta'(0)-\log n+1+\log(2\pi)+O(n^{-1}),
\]
thus cancelling the leading order term $\log n$ from the case of
radius unity spheres. By the above along with Theorem \ref{LaplaceTheorem}, we find the new limit and asymptotics.
\begin{corollary}
On the standard spheres rescaled to unit volume, we have the limit
\[
\lim_{k\to\infty}\det(\Delta, S_\mathrm{rescaled}^{2k+1})=\frac{1}{2\pi e}.
\]
In fact the asymptotics can be displayed as (for $n=2k+1$)
\[
\zeta'_{\Delta, S^{n}}(0)=1+\log(2\pi)+O\Big(\frac{\log\log n}{\log n}\Big).
\]
\end{corollary}

\section{Proof of a stationary phase lemma}
\label{sec:proof}
\begin{proof}[Proof of Lemma \ref{stationaryphaselemma}]
Let $s_a:=\tan^{-1}(e^{-a}), s_b:=\tan^{-1}(e^{-b})$ and compute
\begin{equation}\label{LemmaCompute}
\begin{split}
&\int_a^b\sin\big(n\tan^{-1}(e^{-x})\big)\big(1+e^{-2x}\big)^{-\fracsm{n}{2}}\varphi(x)dx\\
&=\int_{s_b}^{s_a}\sin\big(nx\big)\big(1+\tan^{2}x\big)^{-\fracsm{n-2}{2}}\frac{\varphi\big(\log\big(\frac{1}{\tan x}\big)\big)}{\tan x}dx\\
&=\sum_{k=1}^{n_{ab}}\int_{s_b+\fracsm{(k-1)\pi}{n}}^{s_b+\fracsm{k\pi}{n}}\sin\big(nx\big)\big(1+\tan^{2}x\big)^{-\fracsm{n-2}{2}}\frac{\varphi\big(\log\big(\frac{1}{\tan x}\big)\big)}{\tan x}dx\\
&+\int_{s_b+\frac{n_{ab}\pi}{n}}^{s_a}\sin\big(nx\big)\big(1+\tan^{2}x\big)^{-\fracsm{n-2}{2}}\frac{\varphi\big(\log\big(\frac{1}{\tan x}\big)\big)}{\tan x}dx,
\end{split}
\end{equation}
where $n_{ab}:=\Big\lfloor \frac{n}{\pi}(s_a-s_b)\Big\rfloor$. Thus the last term above can be estimated by
\[
\begin{split}
&\int_{s_b+\frac{n_{ab}\pi}{n}}^{s_a}\Bigg\vert\big(1+\tan^{2}x\big)^{-\fracsm{n-2}{2}}\frac{\varphi\big(\log\big(\frac{1}{\tan x}\big)\big)}{\tan x}\Bigg\vert dx\leq C_{\varphi,a}\Big(1+\tan^2\big(s_a-\frac{\pi}{n}\big)\Big)^{-\fracsm{n-2}{2}},
\end{split}
\]
when $n\geq\fracsm{\pi}{s_a}$, and is $O(n^{-\infty})$ as $n\to\infty$, for any $k$. Note as claimed, that the error term is only uniform in the parameter $b$. The first term in (\ref{LemmaCompute}) is written as a sum over the half-periods of the sine function. Fixing $n$, the alternating behavior gives the estimate
\[
\begin{split}
&\Bigg\vert\sum_{k=1}^{n_{ab}}\int_{s_b+\fracsm{(k-1)\pi}{n}}^{s_b+\fracsm{k\pi}{n}}\sin\big(nx\big)\big(1+\tan^{2}x\big)^{-\fracsm{n-1}{2}}\frac{\varphi\big(\log\big(\frac{1}{\tan x}\big)\big)}{\tan x}dx\Bigg\vert\\
&\leq\int_{s_b}^{s_b+\fracsm{\pi}{n}}\big\vert\sin\big(nx\big)\big\vert\big(1+\tan^{2}x\big)^{-\fracsm{n-1}{2}}\frac{\varphi\big(\log\big(\frac{1}{\tan x}\big)\big)}{\tan x}dx\\
&\leq\frac{\fracsm{\pi}{n}}{\tan s_b}\max_{s_b\leq x\leq s_b+\fracsm{\pi}{n}}\varphi\big(\log\big(\fracsm{1}{\tan x}\big)\big).
\end{split}
\]
Similar computations apply for the contribution from the real part of the integral.
\end{proof}

\begin{acknowledgements} The author would like to thank C. B\"{a}r and S. Schopka for raising the question of limits in their paper \cite{BS}. Thanks goes to Kate Okikiolu, for useful criticism of the manuscript, and to Department of Mathematics, University of Pennsylvania, Philadelphia for having the author as visiting graduate student during the creation of this paper.

The work is partly supported by my Elite Research Scholarship 2006, from The Danish Ministry of Science, Technology and Innovation.
\end{acknowledgements}

\bibliographystyle{amsalpha}

\end{document}